\documentclass[12pt]{article}
\usepackage{ae} 
\usepackage[T1]{fontenc}
\usepackage{amssymb}
\usepackage{latexsym}
\usepackage{amsmath}
\usepackage[ansinew]{inputenc}
\usepackage{euscript}
\usepackage{epsfig}
\usepackage{eurosym}
\usepackage{multirow}
\usepackage{anysize}

\newcommand{\qed}{\hspace*{\fill}
            $\Box$\smallskip}
\newcommand{\thmqed}{\hspace*{\fill}
            $\blacksquare$\smallskip}

\newenvironment{proof}{\noindent {\bf Proof:} \par}
                      {\qed}

 \newtheorem{theorem}{Theorem}
 \newtheorem{corollary}{Corollary}[section]
 \newtheorem{lemma}{Lemma}[section]
 \newtheorem{example}{Example}[section]
 \newtheorem{remark}{Remark}[section]
 \newtheorem{definition}{Definition}[section]
 \newtheorem{proposition}{Proposition}[section]
\begin{document}
\title{Minimum Price in Search Model\footnote{The Author has benefited from insightful comments from Prof. Rady, Prof. Janssen, Prof. Felbermayr and Prof. Holzner}}
\author{Sergey Kuniavsky\footnote{Munich Graduate School of Economics, Kaulbachstr. 45, Munich. Email: Sergey.Kuniavsky@lrz.uni-muenchen.de. Financial support from the Deutsche Forschungsgemeinschaft through GRK 801 is gratefully acknowledged.}}
\maketitle

\begin{abstract}

This paper investigates the effects of a low bound price. To do so, a popular and empirically proven model (Stahl (89') \cite{Stahl89}) is used. The model is extended to include an exogenously given bound on prices sellers can offer, excluding prices below such bound. The finding are rather surprising - when the bound is set sufficiently high expected price offered (EPO) by sellers drops significantly. The result seem to be robust in the parameters of the model, and driven by the information provided to consumers by such legislation step: when the limitation is set at sufficiently high levels all consumers anticipate the bound price, and searchers reject any price above it. As a result sellers offer the bound price as a pure strategy.
\end{abstract}


\section{Introduction}

Despite having free markets, it is often the case that governments intervene in trade. The most common example of such is taxation, where certain goods are taxed. This is done despite the known results on welfare impact of such interventions. However, in some cases, e.g. Alcohol or Tobacco products, the government sees a reason to reduce the consumption by imposing a tax. One of the reasons behind it is noted by Cnossen and Sijbren in \cite{alc_price}, table 8, an impressive literature connects higher price to lower consumption of alcohol, in coordination with common belief, and helps to reduce many of the high costs of Alcohol consumption - over 120 bn Euros in 2003 only.



An additional new step of intervention is suggested by the Scottish government. Scotland proposes a minimum price per alcohol unit (10 ml) of 0,40\textsterling. Meng et al. in a report by Sheffield university \cite{report}, perform a research that concludes that such a step would increase prices on alcohol. Clearly, in cases when this bound falls below the market price it has little effect and it would increase prices when it falls above market price level. As shown in Table 2.1 in Meng et al. \cite{report}, the latter happens in many cases, but not in all. Here we investigate what happens in the interim case, when the minimum price is set at market price level. Namely, some retailers sold products below this price, and similar products were also sold above the limit. Clearly, the interesting case is when the average market price is above the limitation. The prices are taken from English supermarket, however due to the close relation between England and Scotland, and a similar step being planned by the English government, it is a relevant example.

To emphasize, suggested minimum prices would be as follows:
\begin{itemize}
\item Liter of 5\% beer would cost at least 2\textsterling (around 2,5\euro).
\item A 750 ml. bottle of 15\% wine - 4.5\textsterling (around 5,6\euro).
\item Half liter of 40\% Vodka - 8\textsterling (around 10\euro).
\end{itemize}

To add to the relevancy of the question, a product example is provided on figure \ref{beer_joint}: 4x440ml of Bavaria beer 2.8\% lager beer has suggested minimum price is 1.97. On 25.10.2013, prices were varying between 1.79 and 2.99\textsterling. \footnote{Prices are taken from http://www.mysupermarket.co.uk}. Therefore, the suggested limitation falls in the relevant area - some sell above and some sell below, with average price above the limitation.

\begin{figure}
 \centering      
 \includegraphics[width=90mm]{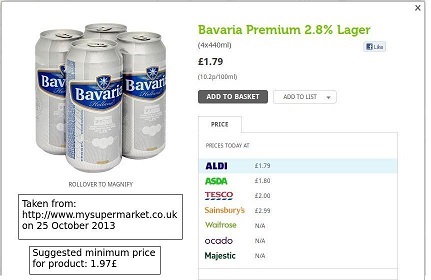}
\caption{Price of 4 pack Bavaria Lager beer from mysupermarket.co.uk}
  \label{beer_joint}
\end{figure}


The existing empiric literature, for example, Stockwell et al. \cite{minprice} and by Meng et al. \cite{report}, imply that introducing minimum price would increase the expected prices of the relevant goods. However, there is one main aspect to look at. The minimum price imposed by the Scottish legislator is not necessarily higher than the previous market price of the good, as noted by Meng et al. \cite{report} table 2.1. In case when the limitation is set above market price, it is clear that prices would rise, as a result of such step. Similarly, if the minimum price is set too low it would have no effect on pricing as there is no need for anyone to adjust price. However, this paper wishes to investigate what happens when minimum price is set moderately. This implies setting the minimum price at some existing market price. The question asked here is what happens in the interim case, where a minimum price is set at price levels which are in equilibrium support, as seen in the example at figure \ref{beer_joint}. The model analyzed here has mixed strategy equilibrium. Therefore, it leaves a range of prices in the relevant interval.

This paper takes a deeper look into introducing lower bounds on price, keeping positive profit for sellers in order to keep them on the market. This is done by analyzing a popular search model, and comparing its equilibrium to one where a lower bound price is introduced. The model analyzed is perhaps one of the most popular search models in the literature - the Stahl Search model, introduced by Stahl \cite{Stahl89}. This simple model has good tractable solution, and therefore has a wide literature. Moreover, Janssen et al. \cite{StahlOkEmp} find that this model indeed predicts product pricing of many goods, and finds that pricing of 86 from the 87 tested products is in line with the Stahl Model. Additionally, Baye et al. \cite{ShoppersExplained} points out that the consumer types suggested by the model are empirically profound. Therefore, such an important model provides a good indicator for checking introduction of lower bounds on prices.

\subsection{Results Overview}

At the first glance, the findings are rather surprising - when the lower bound is chosen at a certain level, not only that prices sink, but selecting the lowest allowed price by all sellers is equilibrium. In an example it reduces the expected price offered (EPO) by approximately 13\%! Most importantly, the minimum price minimizing the EPO depends only on parameters of consumer side, and will not be affected by entry of additional sellers. This would keep the lower bound fixed trough changes on the supply side. Additionally, in the case of two sellers it would be a unique equilibrium. Moreover, informed consumers (shoppers) are indifferent and all the benefits are to the uninformed consumers (searchers). 

Comparative statics suggest that information drives the results. The higher the share of searchers the larger is the difference to the original Stahl model equilibrium. Those receive a credible signal on what are the prices on the market, and take it into consideration when checking prices. Note that since the minimum price ensures positive profit for sellers, it is profitable for sellers to sell already at the minimum price. This would not hold if the minimum price would be replaced with a governmental tax.

When the minimum price is set sufficiently high the structure of the equilibrium is a pure one, at the minimum price level. This implies that the entire market has the same minimum price, in the spirit of Bertrand competition (for more on Bertrand competition see, for example, Baye and Morgan \cite{BertrandEQ}). Minimum price at lower levels will impose a mass point at the minimum price, and some continuous distribution with lower prices than in the original Stahl model equilibrium. The implication is clear - as long as the minimum price is not set sufficiently high or sufficiently low it will reduce market prices, instead of increasing them. Moreover, pricing competition would be reduced since all sellers would use the same pure strategy.

The reasoning behind such behavior is simple. When the limitation is set at sufficiently high levels, but below the expected price in equilibrium, the signal is credible. Searchers believe this limitation plays a significant role, and as a result prices would be at the bound level or close to it. As a result, also the searchers receive a signal on what are the prices in the economy, and expect the bound price to be offered. Thus, searchers set their reserve price sufficiently close to the bound price. Sellers anticipating such a behavior would find it profitable to offer the bound price purely, in order to compete for the shoppers, in a Bertrand competition. Therefore, pure equilibrium at the bound price prevails. When the limitation is set at lower levels this effect diminishes. Now searchers believe that sellers are less probable to be affected by such legislation. As a result, reserve price would be set at a higher level. From sellers' perspective, a mass point on the bound level would be set, with some probability mass above it, but EPO remains below the one in original Stahl model.

As a result the expected price in the market drops when such minimum price is introduced in a market with two or three sellers, and according to Stahl \cite{Stahl89}, also could hold for a larger number of sellers. This is due to a conjecture from Stahl \cite{Stahl89} that reserve price (and thus also EPO) rises with the number of sellers. A table with an example for such increase is provided in Stahl \cite{Stahl89}. The logic for such increase is simple - more consumers reduce the chance to be the cheapest seller, and therefore decrease the motivation to offer lower prices.

This paper suggestion is clear: one should be careful when setting a minimum price, as it can reduce prices and actually increase demand. Consider the example in figure \ref{beer_joint}. If the legislator sets the lowest price at 1.97 \textsterling, searchers would be aware of this limitation. They would expect the price not to be too far above the limitation. Therefore, 2.99\textsterling, offered by Sainsbury's would not be an acceptable price for searchers. The seller anticipating it would need to reduce price closer to the limitation price.


The structure of the paper is as follows: First the Stahl model is introduced, and relevant results on it are provided. Then, minimal price is imposed on the Stahl model, followed by results and an example where prices are lower due to such a step. Afterwards, the specific case with two sellers is looked upon. Lastly, some additional characteristics of the price reducing equilibrium are provided, followed by a short discussion. Larger proofs are shifted to the appendix.

\section{Stahl Model}

The Stahl model, as introduced in Stahl \cite{Stahl89} is formally described below. Notation was adjusted to the recent literature on the Stahl model.

There are N sellers, selling an identical good. Each seller owns a single store. Production cost is normalized to 0, and assume that seller can meet demand. Additionally, there are consumers, each of whom wishes to buy a unit of the good, evaluating it at some high $M$. The mass of consumers is normalized to 1. This implies that there are many small consumers, each of which is strategically insignificant.

The sellers are identical, and set their price once at the first stage of the game. If the seller mixes then the distribution is selected simultaneously, and only at a later stage the realizations take place.

The consumers are of two types. A fraction $\mu$ of consumers are shoppers, who know where the cheapest price is, and they buy at the cheapest store. In case of a draw they randomize over all cheapest stores, uniformly. The rest are searchers, who sample prices. Sampling price in the first, randomly and uniformly selected, store is free. It is shown in Janssen et al. \cite{FirstCosts} that if it is not the case then some searchers would avoid purchase and the rest would behave as in the original model. If observed price is satisfactory - the searcher will buy there. However, if the price is not satisfactory - the searcher will go on to search in additional stores sequentially, where each additional visit has a cost $c$. The second (or any later) store is randomly selected from the previously unvisited stores, uniformly. The searcher may be satisfied, or search further on. When a searcher is satisfied, she has a perfect and free recall. This implies she will buy the item at the cheapest store she had encountered, randomizing uniformly in case of a draw. The searchers have an endogenous reserve price, $P_M$, which determines when they are satisfied.

The consumers need to be at both types (namely, $0<\mu<1$). If there are only shoppers - it is the Bertrand competition setting, (see  Baye and Morgan \cite{BertrandEQ}), and if there are only searchers the Diamond Paradox (Diamond \cite{Diamond}) is encountered, both well studied.

Before going on, make a technical assumption on the model. In order to avoid measure theory problems it is assumed that mixing is possible by setting mass points or by selecting distribution over full measure dense subsets of intervals. This limitation allows all of the commonly used distributions and finite combinations between such.

Additionally, a couple of very basic results are introduced:

\begin{itemize} 
\item Sellers cannot offer a price above some finite bound $M$. This has the interpretation of being the maximal valuation of a consumer for the good.
\item Searchers accept any price below $c$. The logic behind it is any price below my further search cost will be accepted, as it is not possible to reduce the cost by searching further.
\end{itemize}

\subsection{Game Structure}

The game is played between the sellers, searchers and the shoppers. The time line of the game is as follows:

At the first stage, sellers set their pricing strategies and simultaneously searchers set their common reserve price. Then, if some sellers used a mixed strategy, realization of mixed strategies is taking place. At the second stage shoppers observe all the realized prices and purchase item at the cheapest store. At the third stage, searchers sample a price randomly (uniformly) selected. If a searcher is satisfied she purchases the item. If not - she pays $c$ and observes another price, and so forth until ether she is satisfied or sampled all prices. When the searcher observed all stores and observed only unsatisfactory prices she would buy at the cheapest store encountered.

When reserve price and pricing strategies are being determined the knowledge of the various agents of the game is as follows:
\begin{itemize}
\item Sellers have beliefs regarding the reserve price set by searchers
\item Searchers have beliefs about which pricing strategies were actually played by the sellers.
\item Shoppers will know the real price in each store in the moment it is realized.
\end{itemize}

The probability that seller $i$ sells to shoppers when offering price $p$ is denoted $\alpha_i(p)$. Let $q$ denote the expected quantity that seller $i$ sells when offering price $p$. The expected quantity sold by a seller consists of the expected share of searchers that will purchase at her store, plus the probability she is the cheapest store multiplied by fraction of shoppers ($\mu$), and is also the expected market share of the seller.

Note that the reserve price ensures that the searcher will purchase at the last visited store, unless all stores were searched.

\subsection{Utilities and Equilibrium}

As the sellers set their price at the start of the game, they examine the expected utility that would be obtained due to their pricing strategy. Searcher decision whether to purchase the item or search further is taken on the later step of the game, when exact utility is available. Therefore, sellers have an ex ante expected utility, whereas searchers and shoppers have ex post utility, as follows:

\begin{itemize}
\item Seller utility is price charged multiplied by expected quantity sold.
\item Consumer utility is a large constant $M$, from which item price and search costs are subtracted.
\end{itemize}


The NE of the game has a Bayesian structure, and 
is as follows:
\begin{itemize}
\item Searchers have a reserve price ($P_M$). 
\item Searchers beliefs coincide with seller strategies played.
\item Seller have beliefs regarding the reserve price which coincide with $P_M$.
\item Reserve price is rational for the searchers
\item No seller can unilaterally adjust the pricing strategy and gain profit in expected terms.
\end{itemize}

\begin{remark}
As the sum of the searcher and seller utilities may differ only in the search cost, any strategy profile where all searchers always purchase the item at the first store visited is socially optimal.
\end{remark}

The original article also pins down the unique symmetric equilibrium of the model:

\begin{theorem}{Stahl, (89') \cite{Stahl89}}\\
In the model described above exists a unique symmetric NE, where all sellers use a continuous distribution function $F$ on the support between some $P_L$ and $P_M$, which is the reserve price.
\end{theorem}

An additional important result involves the equilibrium price distribution. Since a seller is indifferent between all strategies in support when mixing she must have equal profit for all prices she offers. Thus in the original model equilibrium the profit equality implies:
\begin{equation*}
P\left( (1-F(p)^{n-1}) \mu + \frac{1-\mu}{n} \right) = P_M\frac{1-\mu}{n}
\end{equation*}

From here follows that:
\begin{equation}\label{distr}
F= 1 - \sqrt[n-1]{\frac{1-\mu}{n\mu}(\frac{P_M}{P}-1)}
\end{equation}
Since $F(P_L)=0$ we get that:
\begin{equation*}
\frac{1-\mu}{n\mu}(\frac{P_M}{P_L}-1) = 1
\end{equation*}
Which in turn implies that:
\begin{equation*}
P_L = P_M \frac{1+(n-1)\mu}{1-\mu}
\end{equation*}

An additional connection exists between the reserve price $P_M$ and the distribution $F$:
\begin{lemma}
In a seller symmetric equilibrium the reserve price is exactly $c$ above the expected price of $F$.
\end{lemma}
\begin{proof}

Firstly note that in a symmetric equilibrium observed price does not reveal any additional information about prices in other stores. Therefore, the decision whether to search further or not depend only on the believed expected price of sellers. If this expected price is sufficiently lower than observed prices - search will go on. If not - the searcher would be satisfied. Most of the search literature remains in the symmetric world, and therefore, we look on symmetric equilibria.

If $P_M-E(F)$ would be above $c$ - searchers would not be satisfied with price offers of $P_M$, and search further. If it would be below $c$ - searchers would be better off accepting prices up to $E(F)+c$. Since sellers know this, they would offer $E(F)+c$ and expect the searchers to accept such price. The claim of searchers to reject prices below $E(F)+c$ is not credible.
\end{proof}

One can be surprised regarding such equilibrium structure, and mixed seller strategies. The economic intuition behind it is quite simple. Low prices have high probability to attract shoppers, and will have a higher market share. Higher price is attractive to extract information rent from searchers. In equilibrium these two motivations have equal weight, and therefore a mixed equilibrium prevails.

\section{Minimum Price}

Up to here we were discussing the original Stahl model with its known results. Now we wish to introduce a price limitation in the model.
Consider a legislator, who wishes to introduce a minimum price, below which sales are forbidden. For that let us consider the implications of such a step in the Stahl search model.

Let us denote the Stahl model with a minimum price limitation as the limited model. Additionally, fix $c, \mu$ and $N$. Let us reserve to $P_M$ as the reserve price in the corresponding original Stahl model equilibrium, and $F$ the equilibrium distribution. 

\begin{definition}
Let us define the price $P^*$ as follows:
\begin{equation*}
P^* = c \frac{1-\mu}{\mu}
\end{equation*}
\end{definition}

\begin{theorem}\label{pstar_thm}
Consider the Stahl model with the parameters $N,c$ and $\mu$. Let $P_C$ be a price weakly above $P^*$. Suppose that a minimum price of $P_C$ is imposed. In the limited model exists a unique pure strategy equilibrium, where all sellers select $P_C$ purely, and $P_C+c$ is the searcher reserve price.
\end{theorem}

Proof shifted to the appendix.

\begin{remark}
Note that any equilibrium where $P_C>P_L$ must have mass points at $P_C$. The reason is simple - without mass points at $P_C$ there is no motivation to go below this price, as it already attracts shoppers with certainty. Therefore, such equilibrium would prevail also without a cutoff price. However, we know that such equilibrium must have $P_L$ in support which is not possible.
\end{remark}

This theorem provides the first important result. If the minimum price is set sufficiently high we will receive a pure equilibrium where all sellers select reserve price purely. Searchers treat such a signal as important for the market and anticipate that prices should be around the minimum price. Sellers anticipating it, and therefore do not offer prices higher than $P_C+c$. However, the share of shoppers is sufficiently high to make $P_C$ more attractive than $P_C+c$, making a pure equilibrium. This is since at $P_C$ also shoppers would purchase the item in my store, and at $P_C+c$ only searchers originally visiting my store would buy there.


\subsection{Lower Prices}

An important question is whether price offers are higher or lower due to this limitation. This question is dealt with next.

\begin{definition}
Let us denote the ratio between the expected price offered (EPO) in the original model and $P^*$ as $\beta$. If $\beta >1$ the expected price offered in the limited model is lower than in the original model.
\end{definition}

The following lemma suggests a condition which ensures that the expected price offered in the limited model.

\begin{lemma}\label{beta_above_1_lemma}
The expected price offered in the limited model with minimum price of $P^*$ is lower than the one in the original model iff $P_M > c / \mu$.
\end{lemma}

\begin{proof}
Note that the expected price offered in the original model is given by $P_M-c$, and $P^*=c \frac{1-\mu}{\mu}$. Therefore:
\begin{equation*}
\beta=\frac{\mu(P_M-c)}{c(1-\mu)}
\end{equation*}
The expression is larger than one iff:
\begin{equation*}
\mu(P_M-c)>c(1-\mu)
\end{equation*}
after simplifying we obtain the required condition.
\end{proof}

Note that this condition involves an endogenous parameter $P_M$. When it could be explicitly found, a different, exogenous condition could be provided. However, this is a condition which may be true or false, depending on the parameters of the model. As noted by Stahl in \cite{Stahl89}, there is a monotone connection between $P_M$ and $\mu$. Lower $\mu$ (less shoppers) imposes higher reserve price, as the motivation to attract shoppers decrease.

As a conclusion, if in an equilibrium of the original Stahl Model it is the case that $P_M(N,\mu,c) > c/\mu$ then it is possible to reduce the payment of searchers by introducing a minimum price of $P^*$. Searchers sample a store randomly, and in expected terms will observe the expected price sellers offer. Therefore, searchers benefit from such a step. 

Note that if this condition does not hold, the reserve price is rather low in comparison to $P^*$, as $P^*$ would be above $P_M-c$. Only then any price limitation which can be possible would not reduce EPO. Therefore, initially the information rent was rather low. 

Below a numeric example with 2 sellers is provided. It is suggested in Table 1 in Stahl \cite{Stahl89} that the reserve price increases with $N$. Therefore, the case with 2 sellers probably has the lowest reserve price. The reason behind such a motivation is simple - when more sellers are in competition, it is harder to be the cheapest seller decreases. Therefore, there is less motivation to offer discounts, which drive the prices up.

\subsection{Example}

\begin{example}
Consider the Stahl model with 2 sellers, search costs of $c$ and $\mu=1/3$.
\end{example}

In the example $\mu=1/3$, such that $(1-\mu)/N=\mu$. Original model equilibrium distribution is $F(p)=2-P_M/P$ and $P_L=P_M/2$. The expected price of a seller would be then $E=P_M(ln2)$. Since $E+c=P_M$ we get that $c=P_M(1-ln(2)) \approx 0.3 P_M$.

Comparing the expected price offered and $c(1-\mu)/\mu=2c$ we get that such an equilibrium would be more profitable for searchers, as prices offered in the equilibrium with minimum price set at $c(1-\mu)/\mu$ are lower than $E$.

Applying the definition of $\beta$ it is visible that expected price offered is about 13\% less than in the original model equilibrium:
\begin{equation*}
\beta = \frac{P_M-c}{c (\frac{1-\mu}{\mu})} + \frac{P_M \ln(2)}{2P_M(1-\ln(2)} \approx \frac{69}{61} \approx 1.13
\end{equation*}

Therefore, in this case the condition holds and the expected price offered is indeed lower, and is lower by about 13\%. 

\section{Two Sellers}

In the case of two sellers it is possible to provide some additional results. The main reason behind it is the ability to calculate the expected value of original model equilibrium distribution $F$. Remember from equation (\ref{distr}), that in the general case it is:
\begin{equation*}
F= 1 - \sqrt[n-1]{\frac{1-\mu}{n\mu}(\frac{P_M}{P}-1)}
\end{equation*}
And there is no general explicit expression for the expected value of $F$. However, in the case of two sellers, $F$ looks as follows:
\begin{equation*}
F= 1 - {\frac{1-\mu}{2\mu}(\frac{P_M}{P}-1)}
\end{equation*}

In such case the expected value is given by:
\begin{equation*}
P_M \frac{1-\mu}{2\mu} \ln\left( {\frac{1+\mu}{1-\mu}} \right) 
\end{equation*}

Let us denote $\ln{\frac{1+\mu}{1-\mu}}$ as $\kappa$.

Since $E(F)+c=P_M$, we can obtain a value for the reserve price:
\begin{equation}\label{PM_2_explicit}
P_M = c \frac{2\mu}{2\mu-(1-\mu)\kappa}
\end{equation}

From here it is possible to obtain several additional results. Firstly, a lemma shows that for any two seller Stahl model introducing a low price bound would reduce prices.

\begin{lemma}\label{2sellercond1}
The condition from lemma \ref{beta_above_1_lemma} holds for the case of two sellers.
\end{lemma}

Proof shifted to appendix.

\begin{remark}
Stahl \cite{Stahl89} argues in table (1) that increasing the number of sellers increases the reserve price (for fixed $\mu, c$). If this is true, then introducing $P^*$ will have the positive implication for searchers, and the pure equilibrium will hold.
\end{remark}

\begin{proposition}\label{unique_2sel}
Consider a two seller Stahl model, with a limitation is imposed at price $P^*$. In such case the equilibrium, as shown in theorem \ref{pstar_thm} is unique, and no additional mixed or asymmetric equilibria exist.
\end{proposition}

Proof of proposition is shifted to the appendix.

A very important question is whether shoppers benefit or not from such equilibrium. The surprising answer is that when the lowest possible they are indifferent when $P^*$ is selected.

\begin{lemma}
Let $N=2$. Suppose sellers mix independently. Then the expected offer the shoppers observe in orig. equilibrium is exactly equal to $P^*$.
\end{lemma}
\begin{proof}
This follows directly from the expected value of minimum value from two iid variables with the distribution $F=1 - {\frac{1-\mu}{2\mu}(\frac{P_M}{P}-1)}$. Remember that $\min(F_1,F_2)$ is distributed with $1-(1-F)^2$.

Then, the expected price shoppers' encounter is given by the distribution:
\begin{equation}
G(p)=1-(\frac{1-\mu}{2\mu}\frac{P_m}{p}-1)^2
\end{equation}
The expected value of such a distribution is:
\begin{equation}
E(G)=2(\frac{1-\mu}{2\mu})^2P_M(\frac{2\mu}{1-\mu}-\log(\frac{1+\mu}{1-\mu}))
\end{equation}
Setting the value of $P_M$ from equation (\ref{PM_2_explicit}) into the last equation leads to $E(G)=c\mu/(1-\mu)$.
\end{proof}

The two seller case allows us to perform some additional analysis. When $P_C < P^*$ we have equilibria of a different form, as given by corollary \ref{pcbelow}. There, sellers set a mass point with mass $\rho(P_C)$ on the minimum price $P_C$, and a continuous distribution over a parameter and $P_C$ dependent interval $(P_N,P_M)$ where $P_N$ is strictly larger than $P_C$. An important result for comparative statics is below, explaining how does the reserve price (and thus also the expected price offered which is $c$ below $P_M$) change:

\begin{lemma}\label{pmrholemma}
In such equilibria $P_M$ is decreasing in $\rho$.
\end{lemma}

Proof shifted to the appendix.

Thus, increasing $P_C$ between $P_L$ up to $P^*$ increases $\rho$ and slowly decreases the EPO, up to a level where it reaches the lowest value at $P^*$. Higher values keep the pure equilibrium but with a higher EPO, as sketched in figure \ref{interval_col}.

\begin{corollary}\label{interval_col}
Consider the Stahl model with two sellers. For any minimum price in the interval $(P_L,P_M-C)$, the resulting EPO would be lower than in the original equilibrium.
\end{corollary}

Following corollary \ref{pcbelow2}, larger mass point at minimum price implies minimum price closer to $P^*$. Since it is strictly monotone, the equilibrium would be unique.
Note that there will be no jump at $P_L$, due to the continuous nature of the change from the original equilibrium when $\rho$ is small.
Now, it is possible to sketch equilibrium behavior as a function of the minimum price, as done in figure \ref{sketch} for the case of two sellers. Note that the slope between $P_L$ and $P^*$ can be of a different nature, as it is only a sketch.

\begin{figure}
 \centering
\includegraphics[width=100mm]{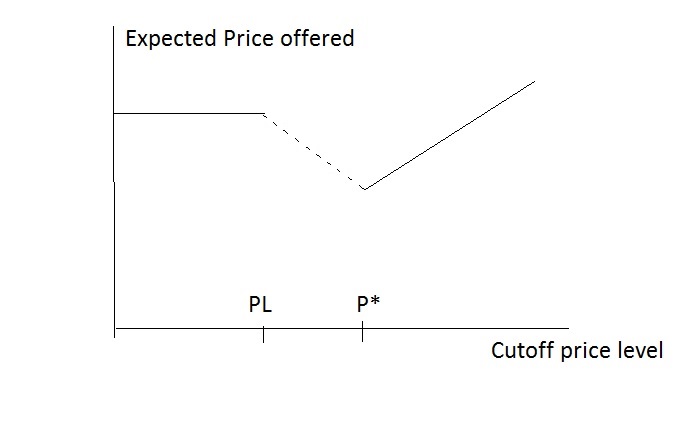}
  \caption{change in EPO as func. of minimum price(2 sellers)}
  \label{sketch}
\end{figure}

\subsection{Two Sellers Summary}

So, for two sellers some strong results are available. Firstly, when limitation is set at $P^*$, it is always a good idea to introduce a minimum price if consumers benefit is before the eyes of the decision maker. The resulting equilibrium is unique. If one is concerned about the shoppers - those are indifferent between the two possibilities and therefore, in total consumers are better off. 

Additionally, no matter where the minimum price is set, EPO would drop. It is not only the specifically picked $P^*$, but any price in the range $(P_L,P_M-c)$ that have such an effect. The effect is global over the entire range of prices. Any level of minimum price below original EPO  where it is not possible to use the original equilibrium will not increase market prices. Clearly, it is maximized at $P^*$, and is lower the further we are from it, but it still prevails. When the limitation is set above $P^*$ we will have a pure Bertrand equilibrium, where all sellers set the limitation price. When the limitation is at $p$, where $p<P^*$, then there would be a mass point at $p$ with some additional distribution mass between two prices strictly above $p$ and below original model reserve price.

\section{More than two sellers}

When more than two sellers are involved, the picture starts to be more complex. Firstly, asymmetric equilibria are possible, as noted by Astone-Figari and Yankelevitch \cite{AsymmetricSearch}. The comparison so far was on symmetric equilibria, and additional ones can be completely different than the ones introduced before.

However, even in the symmetric equilibrium difficulties arise. An additional complexity rises when one tries to calculate the expected value of the original model price distribution. Remember that from equation (\ref{distr}):
\begin{equation*}
F= 1 - \sqrt[n-1]{\frac{1-\mu}{n\mu}(\frac{P_M}{P}-1)}
\end{equation*}
Unfortunately, there is no general explicit expression for the expected value of $F$, which would cover all possible number of sellers. Therefore, most of the results cannot be proven for a general number of sellers.

\subsection{Three Sellers}

It is possible to calculate explicitly the expected value denoted in equation (\ref{distr}) when the number of sellers is three.

Using Matlab and calculating the reserve price $P_M$ for the case of 3 sellers, the following expression was obtained:
\begin{equation*}
\frac{E(F)}{P_M}=\arctan{\left(\sqrt{\frac{3\mu}{1-\mu}}\right)} \cdot \sqrt{\frac{1-\mu}{3\mu}}
\end{equation*}

Note that $E(F)=P_M-C$, and from here follows:
\begin{equation*}
P_M = \frac{c}{1-\frac{E(F)}{P_M}}=\frac{c}{1-\arctan{\left(\sqrt{\frac{3\mu}{1-\mu}}\right)} \cdot \sqrt{\frac{1-\mu}{3\mu}}}
\end{equation*}

As we have found, in two sellers case the reserve price was given by:
\begin{equation*}
c = P_M (1- \frac{1-\mu}{2\mu}\ln{\frac{1+\mu}{1-\mu}})
\end{equation*}

Figure \ref{pm3pm2} states the difference between reserve prices for a given $\mu$ with 2 and 3 sellers is provided below. It clearly shows that the reserve price for 3 sellers is higher for any $\mu \in (0,1)$. Additionally, lower $\mu$ increases the difference in reserve price. From here follows:
\begin{lemma}\label{3sellercond1}
For the case of 3 sellers imposing a lower bound on price at the level of $P^*$ would reduce the expected price offered. The percentage prices drop by slightly more than what was with 2 sellers.
\end{lemma}

Fix $\mu \in (0,1)$ and $c$. Let $P_2$ be the reserve price in the corresponding original Stahl model with 2 sellers, and $P_3$ with three sellers. Then, $P_3>P_2$.
This fact is immediate from figure \ref{pm3pm2}. 

Note that since $P^*$ is independent in the number of sellers and expected price offered is $P_M-c$, higher $P_M$ implies more beneficial equilibrium for searchers.

\begin{corollary}
The condition from lemma (\ref{beta_above_1_lemma}) holds also for three sellers.
\end{corollary}
Since the condition is $P_M > c/\mu$, if it holds for $P_2$, it would also hold for $P_3$. This is in line with Table 1 in Stahl \cite{Stahl89}, suggesting that reserve price is rising with the number of sellers.

An additional result available for three sellers is shoppers' indifference:
\begin{lemma}
Consider the original three seller Stahl model. Shoppers expected price is exactly $P^*$.
\end{lemma}
Similar to the two seller case, applying the expected value of minimum of three $F$ distr. variables would yield the result.

Remember that $P^*=c(1-\mu)/\mu$. Thus, we can calculate:
\begin{equation*}
P^* = P_M (1-\arctan{\left(\sqrt{\frac{3\mu}{1-\mu}}\right)})
\end{equation*}

Additionally, the distribution of shoppers price is given by $1-(1-F)^3$. Calculating the expected value would yield $P^*$. \qed



\begin{remark}
One can repeat the exercise with any desired number of sellers (n) and a desired share of shoppers ($\mu$). Then, a numeric calculation, or, perhaps, in some cases even an analytic expression, for the expected value of $F$ can be found. Then it is possible to calculate $P_M$ and verify whether it is above $P^*+c$.
\end{remark}

\begin{figure}
 \centering
\includegraphics[width=140mm]{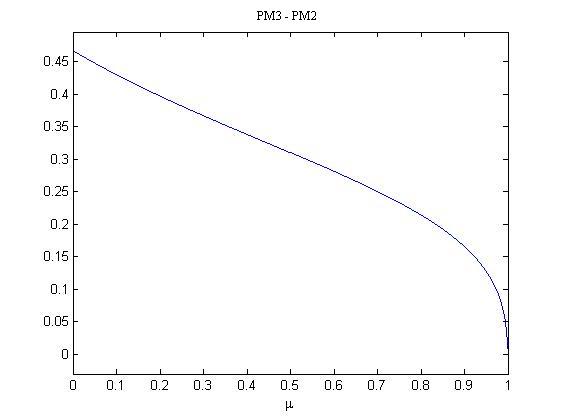}
  \caption{$P_M$ with 3 sellers minus $P_M$ with 2 sellers as func. of $\mu$, in $c$ units}
  \label{pm3pm2}
\end{figure}

\section{Discussion}

A very important question is on the intuition for such a result. A hint can be provided from $\mu$ comparative static. As Figure \ref{pmps} suggests, the lower the share of shoppers, the higher is the effect of introducing such a bound. Additionally, if the price limitation is set \textbf{above} $P^*$, based on theorem \ref{pstar_thm} - the pure equilibrium would still exist and is unique. Therefore, no additional gain for consumers can be obtained. 

Note that if the limitation is set at $P^*$ the information gain of shoppers is exactly zero. They do not get better offers due to their knowledge, but get the same offer as all consumers. Moreover, they get the same offer as before (in expected terms). Therefore, a possible explanation to this phenomenon is information. The law provides additional information to searchers regarding what is cheap and what is not. Using this information searchers form more informed beliefs and get a better deal when purchasing the item.

Additionally, the signal needs to be sufficiently credible. If the bound is set too low, no consumer would believe that sellers would go THAT low on pricing. For example, a low bound set below the support of the original Stahl model would probably have zero effect on results, as all sellers would keep on playing the original model equilibrium. Therefore, this bound needs to be set sufficiently high in order to be credible.

\begin{figure}
 \centering
\includegraphics[width=100mm]{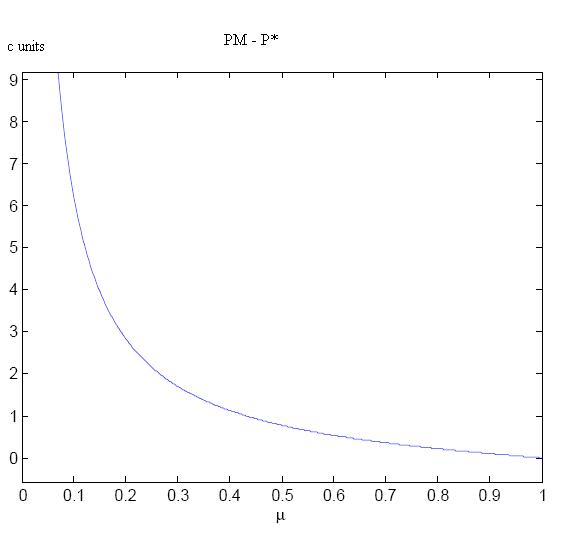}
  \caption{$P_M-P^*$ in $c$ units as func. of $\mu$ (2 sellers)}
  \label{pmps}
\end{figure}

As a result searchers set their reserve price sufficiently close to the bound price. Sellers anticipate it, and since no seller wishes to be above the reserve price, the price dispersion is lower. Additionally, since there is a motivation to attract shoppers and be cheapest, a seller would compare the prices scope available to her. If other sellers have sufficient mass at the bound price, she would not want to deviate from it too. Since the information rent from searchers is low, due to their lower reserve price, it is more attractive to compete for the shoppers.

Therefore, as an outcome we receive the pure equilibrium.

When the bound price is set at a lower level, this effect diminishes and there are mass points at $P_C$. The mass on $P_C$ decreases with the bound. Still, the expected price and reserve price are below the original model equilibrium.

\subsection{Model Relevancy}

The following characteristics are important for the model and results:
\begin{itemize}
\item Homogeneous goods with same production costs
\item Fixed demand which does not depend on the price
\item Information asymmetries among consumers
\item Small number of big players
\item Large price dispersion
\item Reserve price is above $c/\mu$.
\end{itemize}

The first three characteristics are basic for the Stahl model. Firstly, in the Stahl model all consumers end up purchasing the good, no matter the prices. This could reflect a market of essential goods, such as electricity supply, bank account, or as seen by many, cellphone and Internet connection. One may argue that Alcohol is not best described, but as found by Jannsen et al. in \cite{StahlOkEmp}, the Stahl model is significant also for other markets. Additionally, the Stahl model involves a homogeneous good and some asymmetry among consumers. The next point is due to certain results in case of several sellers. This can motivate a covert step to reduce profits in an oligopolistic market. If the market does not have big players or no big ability to fight politicians, then a more harsh approach can be used, as severe taxation. If there are big players with policy influence, then perhaps such a step can be used to reduce their revenue. The next point elaborates on the price reduction. Large price dispersion implies large information rents which would be absent in the pure equilibrium, and imposes also the last, endogenous point. 

The last point is a condition for the new equilibrium to reduce prices, as seem in lemma \ref{beta_above_1_lemma}. Unfortunately, it is not possible to give a general analytic expression to $P_M$, and therefore it involves this endogenous parameter. If this condition does not hold then the lowest possible price reduction is too close to the original model reserve price. Therefore, it is not possible to benefit consumers. However, as shown in lemmas \ref{2sellercond1} and \ref{3sellercond1} it holds for two and three sellers. Note that it may also generally hold, as table 1 in Stahl \cite{Stahl89} suggests that the reserve price increases with $N$. If it is indeed the case, then the condition depicted in lemma \ref{beta_above_1_lemma} will hold for any number of sellers, with similar logic to the one presented in lemma \ref{3sellercond1}.

The strongest result is presented in lemma \ref{interval_col}. The effect is prevailing over an interval of prices. Moreover, it is the maximal possible interval, since minimum price below $P_L$ allow original equilibrium and prices above $P_M-C$ only allow prices above EPO. Therefore, setting a minimum price at a market price level has an opposite effect to a possible original intuition.




\section{Summary}

This paper studies the impact of a lower bound on the price. If such lower bound is introduced, it may possible to reduce price significantly, and the equilibrium would be unique. Clearly, when the price limitation is set below the lowest price on equilibrium it would have no effect, and when it is set above the expected price in equilibrium the prices would be higher. In the intermediate cases the effect is rather surprising at first sight. In the case with 2 sellers, for the entire interval the expected price in equilibrium with minimum price is lower. Thus, when the lowest allowed price set at any relevant price level, expected price offered by sellers gets significantly lower. Still, all sellers have a certain positive, though lower, profit. Therefore, it is not expected that due to such change sellers will go bankrupt. The lowest EPO is obtained when the limitation is set at the level of $P^*$, and it reduces prices for the case of 2 sellers, and, judging by Table 1 in Stahl \cite{Stahl89} probably also with a general number of sellers. 

The price reduction is increasing in the amount of searchers, suggesting that the reason driving the result is providing searchers with valuable information. This helps consumers to define what is cheap, and what can be expected, due to such clear legislation. The fact that at the level price $P^*$ shoppers get the same offer (in expected terms) strengthens this intuition. Lastly, when price limit is set above the critical level of the pure NE, the pure equilibrium still prevails, but with a higher price offered by sellers one to one. Thus, when the information rent is zero, increasing the limiting price has only negative effect on consumers, due to a ban on lower prices.

If the original Stahl model equilibrium uses prices between $P_L$ and $P_M$ (with EPO at $P_M-c$), then, in the case of two sellers setting the minimum price at levels between $P_L$ and $P_M-c$ prices would drop. The drop in prices would be maximized at $P^* = c \frac{1-\mu}{\mu}$. On both sides of $P^*$ it changes continuously. Therefore, one must be careful when setting a minimum price, as it may have an effect opposite to the intention of the legislator. If the limitation is set weakly above $P^*$ the resulting equilibrium would be a pure one, when all sellers set their price at the limitation price. When the limitation is set lower, sellers share a mixed strategy: $P_C$ would have a mass point, and a distribution between two prices strictly above $P_C$. However, EPO still remains below the level of the original equilibrium.

Additional effect of introducing a low bound price is limiting predatory pricing, for example see Snider \cite{Airlines} and Bolton and Scharfstein \cite{Predators}. This would allow additional players to enter the market on the seller side increasing competition. In the real world many markets are oligopolistic and deter new entry, which would not be possible once low bound prices are introduced. This will impose an additional positive impact on many markets, and make them far less oligopolistic. This could add an additional discount in pricing, due to higher number of entrants, which fall beyond the scope of this paper.

\subsection{Future Research}

This paper is, to the author's knowledge, one of the first papers suggesting that imposing a minimum price would reduce prices. Opening a new door often adds many questions and new directions. Some of them are introduced below.

A natural next step is taking these results to the lab, or the real world. Empirically compare offered prices in the original and price limited cases and verify the results here. Then it would be politically possible to offer such measures also in the real world, as today no policy maker would consider such step as price reducing.

Additional theoretic steps that provide interest include comparative statics. How changes in $\mu$ and $N$ affect the results, and price reduction in the most general case. The results here suggest larger benefit when $\mu$ is lower for 2 or 3 sellers, as does shifting from 2 to 3 sellers. General result for any number of sellers would be very useful, yet the general expression may be not analytic. An additional robustness check would be to look at possible asymmetric equilibria of the model, and verify whether in all of them price would be reduced.

An additional important step, is checking the case of heterogeneous searchers, in a model similar to Stahl \cite{StahlHetro}, or heterogeneous sellers, in a similar model to Astorne-Figari and Yankelevich \cite{AsymmetricSearch}. One can examine a completely different search model, for example the one from Varian \cite{Varian}, and verify whether such bounds reduce prices there.

\begin{appendix}
\section{Omitted Proofs}

We begin with the proof for theorem \ref{pstar_thm}.
\begin{theorem}
If prices below $P^*$ are forbidden, exists a unique pure strategy equilibrium (yet possibly some additional mixed ones), where all sellers select $P^*$ purely, and $P^*+c$ is the reserve price.
\end{theorem}

\begin{proof}
The main point in proving the theorem is to show that no seller can deviate. Beforehand note that reserve price of $P^*+c$ is rational, as the expected seller price in equilibrium is $P^*$. Therefore, if a seller offers above $P^*+c$ it is worthy to search on and encounter $P^*$ at the next store. If offered below $P^*+c$ then an additional search is not worthy.

The profit for a seller in equilibrium is as follows:
\begin{equation*}
\pi(P^*) = P^* / N
\end{equation*}

Sellers cannot deviate to a lower price. If a seller deviates to a higher price then there are two possibilities:
\begin{itemize}
\item If the price is weakly below $P^*+c$ then the seller will sell only to searchers initially visiting her store
\item If the price is above $P^*+c$ then nobody will purchase at the store.
\end{itemize}

In the latter case the profit is zero, and in the former case the profit is maximized when the price offered is exactly $P^*+c$. The profit when offering $P^*+c$ is:
\begin{equation*}
\pi(P^*+c)= (P^*+c)(1-\mu)/N
\end{equation*}

Therefore, the deviation is not profitable iff $\pi(P^*) \geq \pi(P^*+c)$, which implies:
\begin{eqnarray*}
P^* / N \geq (P^*+c)(1-\mu)/N \\
P^* \geq (P^*+c)(1-\mu) \\
P^* \mu \geq c(1-\mu)\\
P^* \geq c \frac{1-\mu}{\mu}
\end{eqnarray*}
The last inequality holds due to the definition of $P^*$. Therefore, any minimum price weakly above $P^*$ will also work.

Note that any price other than $P^*$ cannot have positive mass in equilibrium strategy distribution. This is due to undercutting.

\end{proof}

Next the proof for lemma \ref{2sellercond1} is provided.

\begin{lemma}
The condition from lemma \ref{beta_above_1_lemma} holds for the case of two sellers.
\end{lemma}

\begin{proof}

From equation \ref{PM_2_explicit}, it follows that:
\begin{eqnarray*}
f(p) = & \frac{1-\mu}{2\mu}  \frac{P_M}{P^2} \\
E(F(P)) = & P_M-c = \frac{1-\mu}{2\mu} \kappa P_M \\
c = & P_M (1- \frac{1-\mu}{2\mu}\kappa)
\end{eqnarray*}

The required condition is $\mu P_M >c$. This will hold iff:
\begin{equation}
\frac{2\mu - (1-\mu)\kappa}{2\mu} < \mu
\end{equation}

Elaborating the expression yields:

\begin{eqnarray*}
\frac{2\mu - (1-\mu)\kappa}{2\mu}<\mu \\
2\mu - (1-\mu)\kappa<2\mu^2 \\
2\mu - 2\mu^2 < (1-\mu)\kappa
\end{eqnarray*}

This is equivalent to: $\kappa > 2\mu$.
Note that both expressions are 0 when $\mu=0$ and both expressions are increasing in $\mu$ when $\mu \in (0,1)$. Comparing the derivatives one sees that:
\begin{equation}
\kappa' = \frac{1}{1+\mu}+\frac{1}{1-\mu}=\frac{2}{1-\mu^2} > 2 = (2\mu)' \forall \mu \in (0,1)
\end{equation}

Thus $\kappa$ raises more steep than $2\mu$. Since equality is obtained at 0, it follows that in the range $\kappa > 2\mu$.

\end{proof}

Lastly, the proof for proposition \ref{unique_2sel} is provided:
\begin{proposition}
Consider a two seller Stahl model. Suppose a limitation is imposed at price $P^*$. Then the equilibrium, as shown in theorem \ref{pstar_thm} is unique, and no additional mixed or asymmetric equilibria exist.
\end{proposition}

\begin{proof}
Note that the only additional equilibria where a ban on low prices plays a role are ones where sellers set mass points on the price $P^*$.
Let the mass on the price $P^*$ be denoted $\rho$. Additionally, as before no seller would offer a price above the reserve price.

The profit of a seller offering the price of $P^*$ would be:
\begin{itemize}
\item If the other seller sets price $P^*$ - profit of $P^*/2$
\item In any other case profit of $P^*(\mu+(1-\mu)/2)$.
\end{itemize}
Note that due to undercutting at no price except $P^*$ would be a mass point. Additionally, the reserve price would be the highest price in seller strategy support.

From here follows that seller strategy must be an isolated point at $P^*$, and then an interval from some $P_N$ up to $P_M$. This is since at $P_N$, the next offered price above $P^*$ the profit is:
\begin{itemize}
\item If the other seller sets price $P^*$ - profit of $P_N(1-\mu))/2$
\item In any other case profit of $P_N(\mu+(1-\mu)/2)$.
\end{itemize}
Since $(1-\mu)/2 < 1/2$, we get that $P_N > P^*+\varepsilon$ for some positive $\varepsilon$.

From equal profit when mixing one can get the distribution for the interval $(P_N,P_M)$:
\begin{equation*}
F(P)=1-\frac{1-\mu}{2\mu}(\frac{P_M}{P}-1)
\end{equation*}
Let us denote the mass at $P_L$ as $\rho$. From here, $F(P_N)=\rho$, the ratio of $P_N$ and $P_M$ is as follows:
\begin{equation*}
\frac{P_M}{P_N} = \frac{1+\mu-2\mu \rho}{1-\mu}.
\end{equation*}

Combining continuous distribution with mass point, yields the expected value as follows:
\begin{equation}\label{expval_lowpc}
E(F)= \rho P^* + \frac{1-\mu}{2\mu}\log{\frac{P_M}{P_N}}
\end{equation}
Remember that $E(F)=P_M-c$ and $P^* = c \frac{1-\mu}{\mu}$. Thus, once can replace $c$ and $P^*$ with $P_M$ multiplied by corresponding elements. Note that:
\begin{eqnarray*}
P_M(1-\mu)/2 = P^*(1-\rho/2)\mu+(1-\mu/2)\\
P_M = P^* \frac{1+\mu-\mu\rho}{1-\mu}\\
P_M = c \frac{P_M}{P^*}\frac{P^*}{c}
\end{eqnarray*}
Replacing $P^*$ and $c$ with fractions of $P_M$ would cause $P_M$ to be narrowed down throughout the expression, leaving the equality $E(F)=P_M-c$ as follows:
\begin{equation}\label{no_ne_eq}
\frac{1-\mu}{2\mu}(\log(1+\mu-2\mu\rho)-\log(1-mu)) = \frac{1-\rho}{1+\mu-\mu\rho}
\end{equation}
This equation does not have a solution for $\mu, \rho \in (0,1)$. For the value of the expression depicted at equation (\ref{no_ne_eq}) (left side minus right side) see Figure \ref{No_add_eq}.

\begin{figure}
 \centering
\includegraphics[width=80mm]{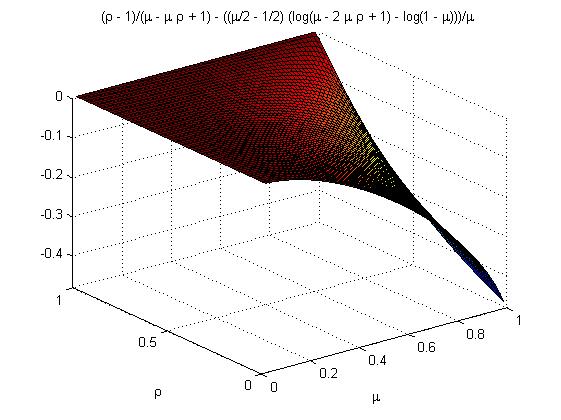}
  \caption{The difference between LHS and RHS of eq. \ref{no_ne_eq}.}
  \label{No_add_eq}
\end{figure}

\end{proof}

and the corresponding corollary:

\begin{corollary}
If lowest price bound set at a level above $P^*$ the pure NE would be unique.
\end{corollary}
The main change is that the ratio between the used $P^*$ and $c$ would be higher. This would increase the right hand side on equation \ref{no_ne_eq}, and it would remain without a solution. This is visible via the logarithm expression than needs now to be even larger than for the $P^*$ minimum price.

\begin{corollary}\label{pcbelow}
In the case that $P_C < P^*$ additional equilibria are possible. For each price for $P_C \in (P_L,P^*)$ exists an equilibrium where sellers have a mass point on $P_C$, and then a continuous distribution over an interval $(P_N,P_M)$ where $P_M$ is the reserve price and $P_N$ is a price strictly larger than $P_C$.
\end{corollary}
When going to the other direction and subtracting form the right hand side the equation \ref{no_ne_eq} would have a solution, implying an equilibrium of the given form. An additional insight is important:

\begin{corollary}\label{pcbelow2}
The larger the $\rho$ satisfying the equation the closer is the minimum price to $P^*$.
\end{corollary}
Follows directly from the fact that the expression is decreasing as $\rho$ goes further from 1.

The last lemma to be proven is \ref{pmrholemma}:
We have shown that in equilibria when $P_C$ is set below $P^*$, $P_C$ has a mass point of $\rho$. Note that lower $P_C$ implies lower $\rho$. 

\begin{lemma}
$P_M$ is decreasing in $\rho$. That is, higher $\rho$ implies lower $P_M$.
\end{lemma}

\begin{proof}
In equation \ref{expval_lowpc} from the appendix the expression for the expected price offered is calculated as a function of $P_M,\mu$ and $\rho$. Suppose for the moment that $P_M$ remains constant. Then, the derivative over $\rho$ is depicted in figure \ref{expval}.
\begin{figure}
 \centering
\includegraphics[width=100mm]{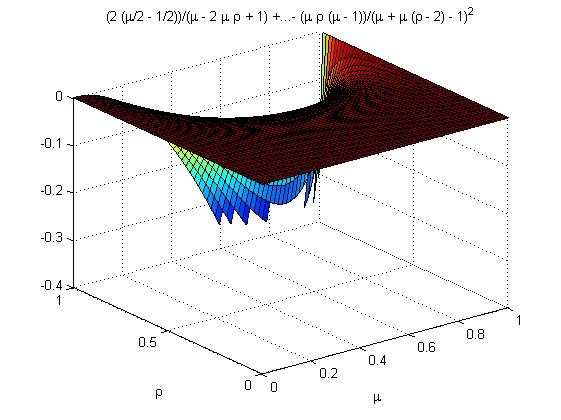}
  \caption{change in EPO as func. of $\mu,\rho$ (2 sellers)}
  \label{expval}
\end{figure}

From figure \ref{expval}, the expected value is decreasing with $\rho$. Remember that $E=P_M-c$. Therefore, if $P_M$ remains constant the difference between $P_M$ and the expected price offered would exceed $c$. From here, $P_M$ must be decreasing in $\rho$.
\end{proof}

\end{appendix}

\end{document}